\documentclass[aps, prd,twocolumn,showpacs,superscriptaddress,nofootinbib]{revtex4}

\usepackage{braket}
\usepackage{amsfonts}
\usepackage{amsmath}
\usepackage{amssymb}
\usepackage{amsthm}
\usepackage{bm}
\usepackage{dcolumn}
\usepackage{epsfig}
\usepackage{graphicx}
\usepackage{graphics}
\usepackage[latin1]{inputenc}
\usepackage{latexsym}
\usepackage{rotating}
\usepackage[colorlinks=true]{hyperref}
\usepackage[usenames]{color}
\usepackage{yfonts}
\usepackage{float}
\usepackage{ucs}
\usepackage{xspace} % Sensible space treatment at end of simple macros
\usepackage{mathrsfs}
\usepackage{booktabs}
\usepackage[toc,page]{appendix}
\usepackage{siunitx}
\usepackage{array}
\usepackage[normalem]{ulem}
\usepackage{tikz}

\newcommand{\be}{\begin{equation}}
\newcommand{\ee}{\end{equation}}
\newcommand{\EH}{{\mbox{\tiny EH}}}
\newcommand{\CS}{{\mbox{\tiny CS}}}
\newcommand{\mat}{{\mbox{\tiny mat}}}

\newcommand{\N}{{\mbox{\tiny Noeth}}}

\newcommand{\q}{Q}
\newcommand{\J}{J}
\newtheorem{theorem}{Theorem}

\newcommand{\B}{\mathcal{B}}

\newcommand{\al}{\frac{\alpha}{4}}

\begin{document}
	\title{Hair Loss in Parity Violating Gravity}
	
	\author{Pratik Wagle}
	\affiliation{eXtreme Gravity Institute, Department of Physics, Montana State University, Bozeman, MT 59717, USA.}
	
	\author{Nicol\'as Yunes}
	\affiliation{eXtreme Gravity Institute, Department of Physics, Montana State University, Bozeman, MT 59717, USA.}
			
	\author{David Garfinkle}
	\affiliation{Oakland University, 	Mathematics and Science Center, 146 Library Drive, Rochester, Michigan 48309, USA}

	\author{Lydia Bieri}
	\affiliation{Department of Mathematics, University of Michigan, East Hall, 530 Church Street, Ann Arbor, MI 48109 USA.}

	\date{\today}
	
\begin{abstract}
		%Abstract character limit is 1300 !!
		%Short abstract below.

		The recent detection of gravitational waves by the LIGO/VIRGO collaboration has allowed for the first tests of Einstein's theory in the \textit{extreme gravity} regime, where the gravitational interaction is simultaneously strong, non-linear and dynamical. One such test concerns the rate at which the binaries inspiral, or equivalently the rate at which the gravitational wave frequency increases, which can constrain the existence of hairy black holes. This is because black holes with scalar hair typically excite dipole radiation, which in turn leads to a faster decay rate and frequency chirping. In this paper, we present mathematical proofs that scalar hair is \textit{not} sourced in stationary, asymptotically flat, and axisymmetric spacetimes, including those appropriate to stars and black holes, when considering extensions of Einstein's theory that break parity in gravity, focusing on dynamical Chern-Simons theory as a particular toy model.This result implies that current electromagnetic observations of accreting black hole systems or of binary pulsars cannot constrain parity violation in gravity today.
	
%PR	OOF 1: 
% non-vacuum, stationary, asymptotically flat and spherically symmetric (because of compactification) spacetime --> scalar charge = 0
% This implies that stationary stars (non-rotating only) have no charge.  

%PR	OOF 2: 
% vacuum, stationary, asymptotically flat and axisymmetric spacetime --> scalar charge = 0. 
% This implies that stationary BHs (rotating or not) have no charge. 
	
\end{abstract}

\maketitle

%%%%%%%%%%%%%%%%%%%%%%%%%%%%%%%%%%%%%%%%%%%%%%%%%%%%%%%%%%%%
%%%%%%%%%%%%%%%%%%%%%%%%%%%%%%%%%%%%%%%%%%%%%%%%%%%%%%%%%%%%
\section{Introduction}

``Who ordered that?!'' This is what Nobel laureate Rabi exclaimed when he was informed about the discovery of the muon particle in 1936. This discovery was made by Anderson and Neddermeyer when carrying out the first studies of cosmic rays as they traversed a magnetic field. We are now entering a new era of physics with the recent observation of gravitational waves by the LIGO and Virgo detectors. The potential for new discoveries is as great as back in the 1930s, and it remains unclear what surprises more sensitive observations will bring.   

Such surprises could extend beyond astrophysics and into theoretical physics. General Relativity (GR) has indeed passed a tsunami of tests in the Solar System~\cite{Will2014}, with radio observations of pulsars~\cite{Stairs2003}, and recently with the first observations of gravitational waves~\cite{PhysRevLett.116.061102}. Solar System tests, however, only probe the weak-field region of the gravitational interaction, while pulsar observations probe the strong but not very dynamical regime. Gravitational waves can sample the extreme gravity region, where the fields are simultaneously strong and dynamical, but such tests are only in the infancy~\cite{Berti:2018cxi}. 
    
At the same time, the theoretical physics community has begun exploring modifications to GR, encouraged both from anomalous observations, as well as theoretical necessity. On the observational side, galaxy rotation curves, among other observations, cannot be explained without some form of dark matter~\cite{doi:10.1146/annurev.astro.40.060401.093923}, while the late time acceleration of the universe seems to require some form of dark energy. Together, these matter-energy content would account for over 95\% of the total matter-energy in the Universe, which has prompted theorists to seek alternative explanations through large-scale modifications to GR~\cite{DeFelice:2010aj,deRham:2014zqa}. On the theoretical side, the intrinsic incompatibility of GR with quantum mechanics has prompted efforts at a variety of unified theories, from string theory~\cite{Polchinski:1998rq} to loop quantum gravity~\cite{Alexander:2009tp,Alexander:2004xd,Taveras:2008yf}.

Whether any of these attempts at modifying GR has anything to do at all with reality will only be determined through experiment and observation. Neutron stars (NS) are ideal laboratories to carry out such tests because of how strong the gravitational interaction is in their vicinity. Radio observations of pulses of radiation emitted by rotating NSs in binary systems have allowed for stringent tests of GR in this strong field regime ~ \cite{Damour:1996ke,Kramer:2006nb,Stairs2003}. In particular, the rate at which the orbital period of the binary decays has been observed to be in agreement with GR to spectacular precision, thus stringently constraining any modification to GR that predicts a different rate ~\cite{Stein:2013wza,Yagi:2012vf,Barausse:2015wia}.  

This single fact has forced us to throw a large heap of theories into the trash bin, but not all theories are automatically ruled out. The reason that the orbital period decay accelerates in many theories is that there is a monopolar scalar or vector field present that activates in the vicinity of isolate NSs. Since the field is anchored to the star, and the star is in a binary, the field moves with the NS, sourcing scalar or vector field waves that remove energy from the system, thus accelerating the orbital decay. Determining whether a monopolar scalar or vector field activates in the presence of NSs is thus paramount to establish the viability of said modified theory.  

A monopolar scalar field is known to be absent in sufficiently symmetric solutions of certain modified theories. For example, non-rotating or slowly-rotating NSs in Einstein-dilaton-Gauss-Bonnet (EdGB) gravity~\cite{Yagi:2015oca} or in dynamical Chern-Simons gravity (dCS)~\cite{Yagi:2013mbt} have no such monopolar field, and in fact, this holds in all shift-symmetric Horndeski theories \cite{Barausse:2015wia}. These theories are all modifications to the GR in which an additional dynamical scalar field is introduced through a non-minimal coupling with a curvature invariant, like the Gauss-Bonnet invariant in the EdGB case or the Pontryagin density in the dCS case.  Black holes (BH), on the other hand, do have a monopolar scalar field in EdGB~\cite{PhysRevD.98.021503}, because certain no-hair theorems can be avoided~\cite{Sotiriou:2015pka}. Radio observations of binary pulsars, however, do not include BHs, and thus, the constraints discussed above can be avoided.  

Can one then prove mathematically that \emph{all} NSs have no such monopolar scalar field in quadratic gravity? The first formal proof focused on EDGB theory~\cite{Yagi:2015oca}, in which the generalized Chern-Gauss-Bonnet theorem is employed for manifolds with boundary~\cite{doi:10.1063/1.531015,Gilkey:2014wca}. This proof therefore applies only for spherically symmetric and stationary spacetimes, in which spacetime can be compactified in the radial and time-like directions into a torus. Recently, this proof was extended to spacetimes without a boundary~\cite{PhysRevD.98.021503} through the use of a generalized Noether theorem~\cite{Wald:1993nt}. The scalar charge was then shown to be given by a certain integral of the Riemann tensor over the inner boundary of the spacetime, which is actually equal to its Euler characteristic.  This Euler characteristic can be thought of as the Euler number for a certain topology, which in turn is related to the number of punctures in the spacetime. 

In this paper, we focus on dCS gravity and prove that in asymptotically flat, stationary, and axisymmetric spacetimes, both in vacuum and not in vacuum, the scalar monopole charge vanishes. First, we work with a one-parameter family of metrics that can be continuously deformed from Minkowski spacetime to a non-punctured spacetime, like that of a NS. Within this framework, we use a variational technique to prove that the scalar charge is given by the (variation of the) integral of the Pontryagin density at spatial infinity, which vanishes identically.  Then, we work with the generalized Noether approach and show that the scalar charge can be related to an integral of the dual Riemann tensor over the inner boundary of the spacetime, which is actually equal to the Euler class of the normal bundle at the boundary. Given that NSs lack such inner boundary, this charge identically vanishes, but moreover, it also vanishes for stationary and axisymmetric BHs, which do have an inner boundary (their horizon) but for which the Euler class of the normal bundle is identically zero.

Since the scalar charge is related to the Euler class of the normal bundle at the boundary, can we understand better what this class is? One way to understand the Euler class is as a measure of how twisted the vector bundle is. If the vector bundle under consideration possesses a nowhere-zero section, then the Euler class vanishes. We will later show that since we work with smooth, non-vanishing vectors at the boundary surface, the Euler class for such a boundary has to be identically zero. Alternatively, one can think of the Euler class in terms of the Euler characteristic. The Chern-Gauss-Bonnet theorem guarantees that the integral of the Euler class over a certain boundary surface equals the Euler characteristic of the manifold. Therefore, one can think of the Euler class as the change in the Euler characteristic per unit surface area. For a BH, the number of punctures in the spacetime is independent of its horizon area, since the latter only depends on its mass and spin, which then implies the Euler class, and thus the scalar charge, are exactly zero.

Our results therefore establish through mathematical proofs that NSs and BHs lack a scalar monopole charge in dCS gravity, preventing binary pulsar observations or gravitational wave observations (of non-spinning binaries) to place constraints on the theory. The remainder of this paper deals with the mathematical details of the proofs summarized above. 
Section~\ref{sec:I} provides a basic introduction to the dCS gravity,
Section~\ref{sec:III} gives an in-depth mathematical proof for the vanishing of scalar charge for a NS whereas Section~\ref{sec:IV} provides the proof for a vanishing scalar charge for BH geometry.
Section~\ref{sec:conclusions} concludes and points to future research.
Henceforth, we adopt the following conventions throughout the paper: we work exclusively in 4-dimensions with metric signature (-,+,+,+), Latin letters (a,b,c,...) in index lists represent spacetime indices, square brackets around indices denote anti-symmetrization, the $\nabla_a$ operator is the covariant derivative and the $\pounds_X$ operator is the Lie derivative with respect to the vector field $X^{\mu}$, the Einstein summation convention is employed unless mentioned otherwise, and we have worked in the geometrical units where $G=1=c$.

%%%%%%%%%%%%%%%%%%%%%%%%%%%%%%%%%%%%%%%%%%%%%%%%%%%%%%%%%%%%
%%%%%%%%%%%%%%%%%%%%%%%%%%%%%%%%%%%%%%%%%%%%%%%%%%%%%%%%%%%%
\section{ \label{sec:I} Dynamical Chern-Simons Gravity} 

This section provides a brief review of dCS gravity and establishes some notation. We only present a minimal review here and direct the interested reader to the recent review paper~\cite{Alexander:2009tp}. The action is given by
\be \label{Action1}
 S = S_{\EH} + S_{\CS} + S_{\vartheta} + S_{\mat}\,,
\ee
where the Einstein Hilbert term is
\be
 S_{\EH} = \kappa \int_{\nu} d^4 x \sqrt{-g} R\,,
\ee
with $\kappa = (16 \pi G)^{-1}$, $R$ the Ricci scalar and $g$ the determinant of the metric tensor $g_{ab}$.
The CS term is
\be
 S_{\CS} =  \frac{\alpha}{4} \int_\nu d^4 x \ \sqrt{-g} \ \vartheta  ~^*\!R R\,,
\ee
where $\alpha$ is a coupling constant, $^*\!R R$ is the Pontryagin density, defined via
\be \label{eq:Pont1}
^*\!R R  = ~^*\! R^{a}{}_{b}{}^{cd} R^b{}_{acd}\,,
\ee
with $^*\!R^{a}{}_{b}{}^{cd}$ the dual Riemann tensor defined as 
\be
^*\!R^{a}{}_{b}{}^{cd} = \frac{1}{2} \epsilon^a{}_{bxy} ~ R^{xy}{}^{cd} \,,
\ee

and $\vartheta$ is a pseudo-scalar field. The action for this field is
\be
S_{\vartheta} = - \frac{1}{2} \int_\nu d^4 x  \sqrt{-g} \left[ g^{ab} (\nabla_a \vartheta) (\nabla_b \vartheta) + 2 V(\vartheta) \right]\,,
\ee
where $\nabla_{a}$ is the covariant derivative operator compatible with the metric, while $V(\vartheta)$ is a potential for the scalar that we set to zero. In addition to these terms, one must also include a matter action that couples only to the metric tensor.

The field equations for dCS gravity can be obtained by varying the action with respect to the metric tensor and the scalar field. These equations are
\begin{align}
	E_{g} = R_{ab} + \frac{\alpha}{\kappa} & C_{ab} - \frac{1}{2 \kappa} \left( T_{ab} - \frac{1}{2} g_{ab}T \right) = 0 \,,
	\label{eq:field-eq}
\\
    \label{eq:CSfield}	E_\vartheta &= \Box \vartheta + \frac{\alpha}{4} ~^*\! R  R = 0 \,,
\end{align}
where $ \Box \equiv \nabla_a \nabla^a$ is the d'Alembertian operator, $T_{ab}$ is the sum of the matter and the scalar field stress-energy tensor, $R_{ab}$ is the Ricci tensor and $C_{ab}$ is the C-tensor, which contains derivatives of the scalar field and the metric. $E_g $ and $ E_\vartheta $ are the gravitational and dCS scalar field equations of motion respectively.

When the scalar field is constant, then dCS gravity reduces identically to GR. This is clear from the field equations, but it can also be seen at the level of the action. The Pontryagin density is the total divergence of the Chern-Simons topological current
\be \label{eq:pont2}
~^*\!R R  = 2 \nabla_a K^a \,,
\ee
where
\be \label{eq:topocurr}
K^a  = \epsilon^{abcd} \Gamma^n_{bm} (\partial_c \Gamma^m_{dn} + \frac{2}{3} \Gamma^m_{cl} \Gamma^l_{dn} )\,,
\ee
$\epsilon^{abcd}$ is the Levi-Civita tensor, and $\Gamma^{n}_{bm}$ is the Christoffel connection. A brief proof of this identify is provided in Appendix A. When the $\vartheta$ field is constant, $S_{\CS}$ becomes a total divergence that can be integrated by parts in the action and eliminated as a boundary contribution.

%%%%%%%%%%%%%%%%%%%%%%%%%%%%%%%%%%%%%%%%%%%%%%%%%%%%%%%%%%%%
%%%%%%%%%%%%%%%%%%%%%%%%%%%%%%%%%%%%%%%%%%%%%%%%%%%%%%%%%%%%
\section{ \label{sec:III}Stellar Hair Loss in Dynamical CS Gravity}

This section presents one of the main results of this paper, i.e.~that scalar charge is not sourced in a non-punctured spacetime in parity violating gravity models like in dCS theory. As we will see below, this result is similar to that of~\cite{Yagi:2015oca}, except that that paper considered a different modified theory and the proof here is different.

\begin{theorem}
Consider a 4-dimensional manifold $M$ that is 
$\mathbb{R}^4$, and is endowed with a metric $g$ of Lorentzian
signature, which is stationary, asymptotically flat and  axisymmetric~\cite{Wald:1984cw,Stewart:1990}.
Furthermore, require that the Riemann curvature tensor be continuous
\emph{almost everywhere}, with discontinuities only in a spatially compact set of
measure zero. In an asymptotically Cartesian coordinate
system, the components of this tensor are required to decay at least as
$\mathcal{O}(r^{-2})$. Consider now a real scalar field $\vartheta$,
stationary under the same isometries of the metric, and that
is governed by a linear coupling in the action to the
Pontryagin density, thus satisfying the equation of motion
\begin{equation}
\label{eq:vartheta-prop-GB}
  \square\vartheta = - \frac{\alpha}{4}~ {}^{*}R{}R\,.
\end{equation}
for some constant $\alpha$.
Then, the asymptotically $1/r$, spherically-symmetric scalar hair (the \emph{scalar charge}) vanishes.
\end{theorem}

\begin{proof}  Let us begin by defining the scalar charge of the scalar field by considering the behavior of $\vartheta$ near spatial infinity $i^{0}$. Stationarity, axisymmetry, and asymptotic flatness imply that the first term in the expansion for the scalar field about spatial infinity takes the form
\be
\vartheta = \frac{F(\theta)}{r} + {\cal{O}}(r^{-2})\,,
\ee
where $\theta$ is the polar angle and $r$ is the radial coordinate far from the source. If $F(\theta)$ is not a constant, then the left-hand side of Eq.~\eqref{eq:vartheta-prop-GB} would contain a term that scales as $r^{-3}$ near spatial infinity, arising from the flat space piece of the curved D'alembertian operator.  However, all terms on the right-hand side of this equation, as well as those left on the left-hand side, decay faster than $r^{-3}$ near spatial infinity. Therefore, under the assumption that $F(\theta)$ is not constant,  Eq.~\eqref{eq:vartheta-prop-GB} cannot be satisfied, and we are then forced to require that $F(\theta)$ be a constant, leading to
\be \label{vartheta}
 \vartheta = \frac{\mu}{r} + \mathcal{O}(r^{-2})\,.
\ee
where $\mu$ is the (constant) scalar charge.

%
%The solutions to Eq. \eqref{eq:vartheta-prop-GB} is a combination of homogeneous and particular solutions. The particular solution decays at least proportional as $r^{-2}$ near spatial infinity, since the Riemann tensor decays at least as $\mathcal{O}(r^{-2})$ by the assumptions of the theorem. The homogeneous solution satisfies the axisymmetric Laplace equation, and thus, it is a function of Legendre polynomials $P_{l} (cos \theta)$ and $r$. Imposing boundary conditions, the homogeneous solution can be expressed as}}
%\be
%\label{eq:homg-sol}
%\vartheta_{hom} = \sum_{\ell=0}^{\infty} P_{\ell} (cos \theta) \frac{a_\ell }{r^{\ell+1}}\,,
%\ee
%{\pw{The slowest decaying piece of $\vartheta$ then corresponds to the $\ell=0$ term of Eq.~\eqref{eq:homg-sol}, which is spherically symmetric, namely}}~\cite{Yagi:2015oca}
%\be \label{vartheta}
% \vartheta = \frac{a_0}{r} + \mathcal{O}(r^{-2})\,.
%\ee
%where $a_0=\mu$ is a constant, which we will call the scalar charge. 

Consider now a one-parameter family of stationary, asymptotically flat metrics $g_{ab}(\lambda)$, and let $\vartheta (\lambda)$ be the corresponding one parameter family of stationary, asymptotically flat solutions of Eq.~(\ref{eq:vartheta-prop-GB}).  We will show that $d\mu/d\lambda =0$ for all $\lambda$ and therefore that $\mu$ is independent of $\lambda$.  Therefore, since any metric that satisfies the conditions of Theorem 1 can be connected to the Minkowski metric by a one parameter family of metrics, and since the solution of Eq.~(\ref{eq:vartheta-prop-GB}) on Minkowski spacetime is $\vartheta =0$, which (trivially) has zero scalar charge, it then follows that in any metric satisfying the conditions of Theorem 1, the solution of Eq.~(\ref{eq:vartheta-prop-GB}) has zero scalar charge.

	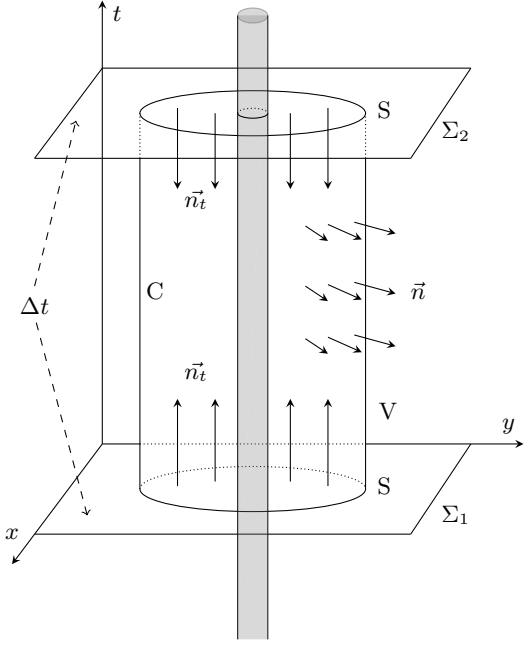
\begin{figure}
	\centering

	\begin{tikzpicture}
	
	\draw (0,0)--(0,4.4) ;
	\draw (3,0)--(3,4.4);
	\draw (1.5,5) ellipse (1.5cm and 0.3cm);
	\draw (0,0) arc(180:360:1.5 and 0.3); 
	\draw [densely dotted] (0,0) arc(0:180:-1.5 and 0.3);
	\draw [->,>=stealth] (-0.5,0.6)--(-0.5,6.5);
	\draw [->,>=stealth](-0.5, 0.6) -- (-1.7,-1);
	\draw (-0.5,0.6)--(0,0.6);
	\draw[densely dotted] (0,0.6)--(3,0.6);
	\draw [->,>=stealth](3,0.6)--(5.1,0.6);
	\draw (-1.4,-0.6)--(3.6,-0.6);
	\draw (3.6,-0.6) -- (4.4,0.6);
	\draw (-0.5,5.6)--(-1.4,4.4);
	\draw (-1.4,4.4) -- (3.6,4.4);
	\draw (3.6,4.4)--(4.4,5.6);
	\draw (-0.5,5.6)--(4.4,5.6);
	\draw (1.3,-2)--(1.3,6.3);
	\draw (1.7,-2)--(1.7,6.3);
	\draw [fill=gray,opacity=0.3](1.5,6.3) ellipse (0.2 and 0.1);
	\fill [gray,opacity=0.3] (1.3,-2)--(1.7,-2)--(1.7,6.3)--(1.3,6.3)arc(180:360:0.2 and 0.1);
	\fill [gray,opacity=0.3] (1.3,-2)--(1.3,6.3)--(1.7,6.3)--(1.3,6.3)arc(180:360:0.2 and 0.1);
	\draw (1.3,5) arc(180:360:0.2 and 0.066);
	\draw [densely dotted](1.3,5) arc(180:360:0.2 and -0.066);
	\draw [densely dotted] (0,4.4)--(0,5);
	\draw [densely dotted] (3,4.4)--(3,5);
	\draw [->,>=stealth] (2.2,3.5)--(2.5,3.3);
	\draw [->,>=stealth] (2.2,2.7)--(2.5,2.5);
	\draw [->,>=stealth] (2.2,2)--(2.5,1.8);
	\draw [->,>=stealth] (2.5,3.52)--(2.95,3.32);
	\draw [->,>=stealth] (2.5,2.72)--(2.95,2.52);
	\draw [->,>=stealth] (2.5,2.02)--(2.95,1.82);
	\draw [->,>=stealth] (2.85,3.55)--(3.4,3.4);
	\draw [->,>=stealth] (2.85,2.75)--(3.4,2.6);
	\draw [->,>=stealth] (2.85,2.05)--(3.4,1.9);
	\node (t) at (-0.3,6) [label=$t$]{};
	\node (x) at (-1.7,-0.9) [label=$x$]{};
	\node (y) at (4.9,0.5) [label=$y$]{};
	\node (n) at (3.7,2.3) [label=$\vec{n}$]{};
	\node (C) at (0.2,2.3) [label=C]{};
	\node (S1) at (4.2,-0.7) [label=$\Sigma_1 $]{};
	\node (S2) at (4.2,4.4) [label=$\Sigma_2$]{};
	\draw [<-][dashed](-0.7,-0.35)--(-1.4,2.3);
	\draw [->][dashed](-1.4,2.7)--(-0.85,4.9);
	\node (dt) at (-1.4,2.1) [label=$\Delta t$]{};
	\draw [->,>=stealth] (0.5,5.07)--(0.5,3.99);
	\draw [->,>=stealth] (1,5)--(1,3.99);
	\draw [->,>=stealth] (2,5)--(2,3.99);
	\draw [->,>=stealth] (2.5,5.07)--(2.5,3.99);
	\draw [->,>=stealth] (1,0.1)--(1,1.2);
	\draw [->,>=stealth] (0.5,0.04)--(0.5,1.2);
	\draw [->,>=stealth] (2,0.1)--(2,1.2);
	\draw [->,>=stealth] (2.5,0.04)--(2.5,1.2);
	\node (n1) at (0.75,1.2) [label=$\vec{n_t}$]{};
	\node (n2) at (0.75,3.49) [label=$\vec{n_t}$]{};
	\node (v) at (3.3,0.7) [label=V]{};
	\node (s1) at (3.25,-0.3) [label=S]{};
	\node (s1) at (3.25,4.7) [label=S]{};
	
	\end{tikzpicture}
	\label{Fig. 2}
	\caption{A schematic diagram of the spacetime needed in our proof. The time axis points upwards, and spacelike hypersurfaces are orthogonal to it, where we have suppressed one spatial dimension. The dark shaded cylinder represents the world tube of the stellar object under consideration. The region of integration in our proof is the outer cylinder C, which is a 3-ball of radius $r$ over a time interval $ \Delta t$. The spacelike unit normal vector $\vec{n}$ is always outward directed whereas the timelike unit normal vector $\vec{n_t}$ is inward directed.}
\end{figure}

We must then only show that $d\mu/d\lambda =0$.  In the spacetime $(M,{g_{ab}} (\lambda))$ [see Fig. 1], 
choose a spacelike slice $\Sigma_1$ and let the spacelike slice $\Sigma_2$ be $\Sigma _1$ translated forward in time by an amount $\Delta t$ along the stationary Killing field.  Let $S$ be a sphere of some large radius $r$ in $\Sigma_1$ (where in the end we will take $r \to \infty$), and let the cylinder $C$ be $S$ translated along the Killing vector by all values from $0$ to $\Delta t$.  Now let $V$ be the spacetime volume bounded by $\Sigma_1$, $\Sigma _2$ and $C$.  

With this setup in mind, let us now consider the solution $\vartheta (\lambda)$ of Eq.~(\ref{eq:vartheta-prop-GB}) in this spacetime and integrate both sides of that equation over the spacetime region $V$.  Integrating the left-hand side of Eq.~(\ref{eq:vartheta-prop-GB}) and applying Stokes' theorem we obtain
\be \label{eq:Stokes}
 \int_{V} d^4x \sqrt{-g} \; \Box \vartheta  = \int_{\partial V} d\Sigma_a \nabla^a \vartheta\,,
\ee
The boundary $\partial V$ consists of three pieces: the cylinder $C$ and two endcaps that consist of the part of $\Sigma_1$ inside $S$ and the corresponding part of $\Sigma _2$ inside $S$.  The contributions of the two endcaps clearly cancel each other, since stationarity implies that those contributions have the same magnitude, while the orientation of the normal vectors required by Stokes' theorem implies that the contributions have opposite sign.  We then find
\be \label{eq:Stokes2}
 \int_{V} d^4x \sqrt{-g} \; \Box \vartheta  = \int_{C} d\Sigma_a \nabla^a \vartheta\ =- 4 \pi \mu \Delta t
\ee
where we have taken the limit of large $r$ and used asymptotic flatness of the metric and Eq. \eqref{vartheta}.

Now let us use $\delta$ as an abbreviation for $d/d\lambda$.  Then from Eq.~(\ref{eq:vartheta-prop-GB}) and (\ref{eq:Stokes2}) we obtain
\be \label{eq:CSfieldint}
(4 \pi \Delta t) \delta \mu =  \al \; \delta \int_{V} d^4x \sqrt{-g} \; ~^*\!R R   \,,
\ee
Using the definition of the Levi-Civita tensor, and the fact that the spacetime metric has Lorentz signature
\be \label{LeviS}
\epsilon^{abcd} =  \frac{-1}{\sqrt{-g}} \tilde{\epsilon}^{abcd}\,,
\ee
where $ \tilde{\epsilon}^{abcd} $ is the Levi-Civita symbol, and the definition of the dual Riemann tensor, we find
\be
(4 \pi \Delta t) \delta \mu = - \al  \; \tilde{\epsilon}^{abmn}  \int_{V}  d^4 x \; R^c_{\phantom{c}dmn} \delta R^d_{\phantom{d}cab}   \,,
\ee
where we used that the Levi-Civita symbol has zero variation because it is independent of the metric. Using the standard expression for the variation of the Riemann tensor in terms of the Christoffel connection, we find
\begin{align}
\label{eq:pre-charge}
(4 \pi \Delta t) \delta \mu = - \frac{\alpha}{2} \; \tilde{\epsilon}^{abmn}
\int_{V}  d^4 x \; R^c_{\phantom{c}dmn}  \nabla_{a} \delta \Gamma^d_{cb} \,.
\end{align}
We now wish to express this equation as a total derivative instead of as a variation. Using the Bianchi identity and the antisymmetric properties of the Riemann tensor and the Levi-Civita symbol,
\begin{align}
\label{eq:covariant}
R^c_{\phantom{c}dmn} \nabla_{a} \delta \Gamma^d_{cb} = \nabla_{a} (R^c_{\phantom{b}dmn} \delta \Gamma^d_{cb}) - \delta \Gamma^d_{cb} (\nabla_{a} R^c_{\phantom{c}dmn})\,.
\end{align}
One can easily see that the last term vanishes upon contraction with the Levi-Civita symbol in Eq.~\eqref{eq:pre-charge}, as it is simply the dualized differential Bianchi identity:
\be
\nabla_{[a} R_{bc]de} = 0 \Rightarrow \epsilon^{abcf} \nabla_{[a} R_{bc]de} \propto \nabla_{a} \;^{*}R^{af}{}_{de} = 0\,.
\ee
We then find
\begin{align}
(4 \pi \Delta t) \delta \mu &= -  \frac{\alpha}{2} \;
\int_{V}  d^4 x \sqrt{-g} \; \nabla_{a} (^{*}R^c{}_{d}{}^{ab} \delta \Gamma^d_{cb}) \,,
\nonumber \\
&= - \frac{\alpha}{2} \;  \int_{C}  d\Sigma_{a} \; {}^{*}R^c{}_{d}{}^{ab} \delta \Gamma^d_{cb}\,,
\end{align}
where in the second line we have again applied Stokes' theorem and used the fact that the contribution of the endcaps cancels.

Let us now evaluate the integral at $C$. By the assumptions of the theorem, the Riemann tensor falls off at least as $ R_{klmn} \sim \mathcal{O} (r^{-2}) $ in an asymptotically Cartesian coordinate system near $i^{0}$, which implies the variation in the Christoffel symbol falls off at least as $ \mathcal{O} (r^{-1} )$. The directed surface element is proportional to $r^{2}$ near $i^{0}$ and thus the integrand decays at least as ${\cal{O}}(r^{-1})$. One then finds that the integral vanishes and thus $\delta \mu = 0$.  This then implies that the scalar charge $\mu$ is independent of $\lambda$, which in turn implies that the scalar charge vanishes identically.
\end{proof}

%The main result that is proved above is that stellar bodies devoid of singularities, like stars or compact objects excluding black holes, have identically zero scalar charge irrespective of whether the bodies are rotating or whether they are moving, and irrespective of their internal composition and equation of state. Our results, therefore, greatly extend those of~\cite{Yagi:2013mbt}, who found numerically that slowly-rotating neutron stars with a given set of equations of state have zero scalar charge. Moreover, our results also extend those of~\cite{Yagi:2015oca} who studied a different theory where the scalar field couples to a parity-even curvature invariant. Combining those results with the results presented here, we can then conclude that \emph{stellar bodies devoid of singularities source no scalar charge in all (parity-even and parity-odd) quadratic gravity theories.}

The main result that is proved above is that rotating (or non-rotating) stellar bodies devoid of singularities, like stars or compact objects excluding BHs, have identically zero scalar charge, irrespective of their internal composition and equation of state in dynamical Chern Simons gravity. Our results, therefore, greatly extend those of~\cite{Yagi:2013mbt}, who found numerically that slowly-rotating NSs with a given set of equations of state have zero scalar charge. Moreover, our results also extend those of~\cite{Yagi:2015oca}, which proved a similar result but in EdGB gravity, a theory in which the scalar field is coupled to a parity-even curvature invariant. Combining those results with the results presented here, we can then conclude that \emph{stationary stellar bodies devoid of singularities source no scalar charge (whether they be rotating or not) in all (parity-even and parity-odd) quadratic gravity theories.}

The implications of the above result spread in a web of interesting truths. When a stellar body sources scalar charge, then its motion forces the emission of dipole scalar waves, in the same way as an electric charge in motion emits dipole electromagnetic waves. Such scalar waves carry dipole energy away from the body, forcing it to slow down. Therefore, when such a stellar body is in a binary system, the emission of scalar radiation forces the binary to lose binding energy more rapidly than predicted in General Relativity. This accelerated loss of energy translates into an accelerated rate of orbital period decay. Since radio observations do not see such a modification from General Relativity in the orbital period of binary pulsars, such dipolar scalar radiation is stringently constrained. But the result presented above says that stellar bodies in quadratic gravity do not source scalar charge, which therefore implies that quadratic gravity cannot be constrained so easily with binary pulsar observations.

%%%%%%%%%%%%%%%%%%%%%%%%%%%%%%%%%%%%%%%%%%%%%%%%%%%%%%%%%%%%
%%%%%%%%%%%%%%%%%%%%%%%%%%%%%%%%%%%%%%%%%%%%%%%%%%%%%%%%%%%%
\section{\label{sec:IV}Black Hole Hair Loss in Dynamical CS Gravity}

This section presents the second main result of this paper, i.e.~that scalar charge is not sourced in a vacuum spacetime in parity violating gravity models like in dCS theory, even if it is punctured by a curvature singularity. 

\begin{theorem} Consider an asymptotically flat, stationary and axisymmetric BH spacetime with a bifurcate Killing horizon (as described by Fig.2) and a real scalar field $ \vartheta $ stationary under the same isometries as the metric satisfying Eq. \eqref{eq:CSfield} . Then the scalar charge as defined in this paper vanishes in such a spacetime.
\end{theorem} 

\begin{proof}

The strategy of the proof is as follows. We will first find a particular antisymmetric tensor field $Q^{ab}$ that satisfies ${\nabla _a}{Q^{ab}}=0$.  One can think of this in analogy to Maxwell's source-free electrodynamics, in which the Maxwell tensor is antisymmetric by definition, it has an associated charge, and it must be divergence free. The latter is a statement of flux conservation, which implies the charge associated with the Maxwell tensor must be conserved on any 2-sphere. Similarly, in our case the antisymmetric $Q^{ab}$ tensor defines, through Gauss's law, a charge for any $S^2$, and since $Q^{ab}$ is divergence-free, this charge is the same for any $S^2$ (See Fig.\ref{Fig. 1}).  We will then find the $Q^{ab}$ charge associated with a two-sphere at infinity, and show that it is proportional to the scalar charge. Subsequently we will find the $Q^{ab}$ charge associated with the bifurcate Killing horizon of the BH and show that it vanishes.  Since these two $Q^{ab}$ charges must be equal, it will then follow that the scalar charge vanishes.
\begin{figure}
	\centering

	\begin{tikzpicture}
	
	\draw (-1.5,-1.5)--(2,2) ;
	\draw (2,2)--(4,0);
	\draw (-1.5,1.5)--(2,-2);
	\draw (2,-2)--(4,0);
	\draw plot[domain=0:2*pi,samples=100](\x*0.202 *pi,{0.25*sin(\x r)}) ;
	\node (b) at (0,0) [label=$\mathcal{B}$]{};
	\node (h1) at (1,1) [label=$\mathcal{H}^+$]{};
	\node (h2) at (1,-1) [label=below:$\mathcal{H}^-$]{};
	\node (s) at (2,0) [label=$\Sigma$]{};
	\node (i) at (4,0) [label=right:$ \infty $]{};
	\draw [densely dashed,->] (3.2,1.2)--(2.7,0.7);
	\node (m) at (3.3,1.1) [label= $\mathcal{M}$ ]{};
	
	\end{tikzpicture}
	\caption{\label{Fig. 1} A schematic diagram for BH spacetime with a bifurcation surface. Here, $ \mathcal{M}$ is manifold and $\Sigma $ defines the Cauchy surface. $\mathcal{H} $ is the bifurcate Killing horizon with the bifurcation surface $\mathcal{B}$. The manifold $\mathcal{M}$ has boundaries at $\mathcal{B} $ and spatial infinity.}
\end{figure}
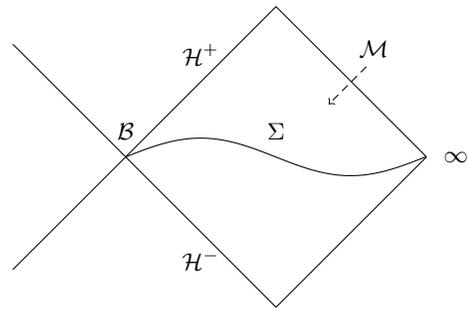

We must then begin by finding the appropriate $Q^{ab}$ and show that it has the properties discussed above. Let us first notice that the only way to construct an antisymmetric tensor of rank 2 from two 4-vectors is through their antisymmetrized exterior (wedge) product. One natural 4-vector to consider is a Killing vector $X^{a}$. Another natural 4-vector is that which arises from Eq.~\eqref{eq:CSfield}, which can be written as ${\nabla _a}{J^a} = 0$ where the current $J^a$ is given by ${J^a}=-{\nabla ^a}\vartheta - (\alpha/2){K^a}$.  We recognize this $J^a$ as the Noether current associated with the shift symmetry of the action under $\vartheta \to \vartheta + c$ ~\cite{doi:10.1063/1.528801}. With this at hand, we might then be tempted to try $Q_{\N,1}^{ab} = - 2 {X^{[a}}{J^{b]}}$, which then leads to 
\be
{\nabla _a}  Q_{\N,1}^{ab} = {X^b}{\nabla _a}{J^a} - {J^b}{\nabla _a}{X^a} - 
{\pounds _X}{J^a} 
\ee 
We know that $J^a$ is divergence-free and that any Killing vector is divergence-free.  Furthermore, any tensor field geometrically defined in terms of the metric is Lie dragged by a Killing vector.  Unfortunately, however, $K^a$ is not a geometrically defined tensor field, so this candidate for $Q^{ab}$ is not quite what we want. 

But not all is lost. The Noether current $J^{a}$ has two pieces, $-\nabla^{a} \vartheta$ and $-(\alpha/2) K^{a}$. Even though the latter is not geometrically defined, the former is and could be used in a new definition of $Q^{ab}$. Choosing then to define $Q_{\N,2}^{ab} = -2  {X^{[a}}{\nabla^{b]} \vartheta}$, however, one discovers rapidly that this antisymmetric tensor is not divergence free: 
\begin{eqnarray} \label{eq:Q1}
{\nabla _a} Q_{\N,2}^{ab} &=& {\nabla _a} ( - 2 X^{[a} \nabla^{b]} \vartheta ) 
\nonumber \\
&=& {X^b} \Box \vartheta - ({\nabla _a}{X^a}){\nabla ^b}\vartheta - {\pounds _X} ({\nabla ^b} \vartheta )
\nonumber
\\
&=& - {\frac \alpha 4}  {X^b} \; {}^*\!R R  \;,
\end{eqnarray}
where the second and third terms in the second line of Eq.~\eqref{eq:Q1} vanish for the reasons described earlier. Clearly then, we must add a \emph{counter-term} ${\cal{C}}^{ab}$ to the new definition of $Q^{ab}$, namely
 \be
 \label{eq:Charge}
Q^{ab} = -2  {X^{[a}}{\nabla^{b]} \vartheta} + {\cal{C}}^{ab}\,,
\ee
such that the divergence does vanish. 

The counter-term must then satisfy 
\begin{align}
\label{eq:no-div-zero}
\nabla_{a} {\cal{C}}^{ab} = \frac{\alpha}{4} X^{b}  \; {}^*\!R R\,.
\end{align}
One can verify by direct evaluation that the counter-term we need is
\begin{align}
{\cal{C}}^{ab} = \frac{\alpha}{2} \; {}^*\!R^{abcd} \; \nabla_c X_d\,.
\end{align}
We can prove this statement by taking the divergence of the counter-term:
\begin{align}
{\nabla _a} {\cal{C}}^{ab} &= 
\frac{\alpha}{2} ({\nabla _a}  {}^*\!R^{abcd}) (\nabla_c X_d) 
+ \frac{\alpha}{2} \; {}^*\!R^{abcd} \; {\nabla _a}{\nabla _c}{X_d}\,,
\end{align}
but the first term on the right-hand side vanishes by the Bianchi identity, while the second term simplifies using the definition of the Riemann tensor in terms of the commutator of covariant derivatives, leading to
\begin{eqnarray}
\nabla_{a} {\cal{C}}^{ab} = - \frac{\alpha}{2} \; {}^*\!R^{abcd} \; {R_{cdae}}{X^e}\,.
\end{eqnarray}
This expression is still not quite the same as Eq.~\eqref{eq:no-div-zero}, mainly because the free index in not on the Killing vector. We can, however, use the following expression for the Riemann tensor 
\be
R_{bcde} = -\frac{1}{4}\epsilon^{mnxy} \epsilon_{demn} R_{bcxy}\,,
\ee
and then contract one of the Levi-Civita tensors above with the Levi-Civita tensor implicit in the dual Riemann through the Levi-Civita Kronecker identity to find identically Eq.~\eqref{eq:no-div-zero}. One can arrive at a similar conclusion using methods from exterior calculus as done in a different theory in~\cite{PhysRevD.98.021503}, which we spell out in Appendix \ref{Cartan}.

With the expression for the conserved charge in Eq. \eqref{eq:Charge}, let us now consider an asymptotically flat, stationary and axisymmetric punctured spacetime with a stationary and axisymmetric dCS field $ \vartheta $ that satisfies the dCS field equations assuming all fields are smooth throughout the manifold ${\cal{M}}$. We assume this spacetime possesses a bifurcate Killing Horizon $\mathcal{H}$ with a bifurcation surface $\B$ that is compact. An important property to remember is that the extrinsic curvature of the bifurcation surface vanishes. We choose $X^a$ to be the particular Killing field that is tangent to the null generators of the horizon: $ X^a = t^a + \Omega_{\mathcal{H}} \phi^a $ where $ t^a $ denotes the time translation Killing field, $ \phi^a $ is the axial Killing field and $\Omega_{\mathcal{H}}$ is the angular velocity of the horizon. This Killing field also vanishes at the bifurcation surface.  Consider now a Cauchy surface $ \Sigma $ denoting the hypersurface that extends from $\B$ out to spatial infinity, with the assumptions that the boundary at spatial infinity is asymptotically flat with the asymptotic conditions 
\begin{align} \label{eq:AsymCondition}
t^\mu & = (\partial_t)^\mu + \mathcal{O}(1/r) \,, \\
\phi^\mu & = (\partial_\phi)(1+\mathcal{O}(1/r)) \,, \\
g_{\mu \nu} & = \eta_{\mu \nu} + \mathcal{O}(1/r) \,, \\ 
\vartheta & = \frac{\mu}{r} + \mathcal{O}(1/r^2) \,,	
\end{align}
 where $ \eta_{\mu \nu} $ is the Minkowski metric and all the n-th derivatives of the above quantities are required to fall off as $ 1/r^n $. 

Using Eq.~\eqref{eq:Charge}, we will now prove our main result, i.e.~that the scalar charge vanishes on such a BH spacetime. The charge associated with the two-sphere at infinity is
\begin{align} \label{eq:scalarcharge}
\int_\infty u_a ~ d \Sigma_b {Q^{ab}} =& -2 \int_\infty u_a ~ d \Sigma_b X^{[a}\nabla^{b]} \vartheta 
\nonumber \\
&= - \int_\infty d \Sigma^r \partial_r \vartheta 
\nonumber \\ 
&= 4 \pi \mu \,.
\end{align}
Here the two-sphere is a sphere of radius $r$ in an asymptotically flat spacelike surface $\Sigma$ and $u^a$ is the unit timelike normal to $\Sigma$.  The curvature term from Eq.~\eqref{eq:Charge} does not contribute as the Riemann tensor falls off at least as $1/ r^3 $ at spatial infinity, whereas the $\phi^a$ contribution coming from the Killing field vanishes as $ \phi^a $ is tangent to the 2-sphere at infinity. 

On the bifurcation surface $\mathcal{B}$, the Killing vector vanishes $ X^a |_\mathcal{B} = 0 $, and its covariant derivative is proportional to the binormal $\tilde{n}_{a}{}^{b}$, namely $ \nabla_a X^b |_\mathcal{B} = \kappa \, \tilde{n}_a{}^b $~\cite{Wald:1984cw}, where $ \kappa $ is constant \cite{KAY199149}. Thus, the charge associated with the bifurcation surface is 
\be \label{eq:final}
\int_\B u_a ~ d\Sigma_b {Q^{ab}} = \frac{\alpha}{2} \kappa \int_\B u_a ~ d\Sigma_b \; {}^*\!R^{abcd} ~ \tilde{n}_{cd} \,.
\ee
Let us now consider the pullback of $ \tilde{n}^{cd} R_{abcd} $ to $\B$. We can rewrite $ \tilde{n}^{ab} $ in terms of null normals to $\B$: $ p^a $ and $ q^b $ where $ p_a q^a = -2 $ and $ \tilde{n}^{ab} = p^{[a}q^{b]} $. Then, 
\be
\tilde{n}^{cd} R_{abcd} = p_c \nabla_{[a} \nabla_{b]} q^c = \nabla_{[a|}(p_c \nabla_{|b]} q^c) - \nabla_{[a|}p_c \nabla_{|b]} q^c,
\ee
but since the extrinsic curvature is just the tangential part of the covariant derivative of the null normals, and the extrinsic curvature vanishes on $\B$, it follows that the second term above vanishes. This is just the geometric statement that the null normals are parallel transported along $\B$. 
%Thus, the pull-back to $\B$ gives $ \nabla_c p^d = \chi_c p^d $ and $ \nabla_c q^d  = - \chi_c q^d $ for some one-form $ \chi_c $, which then implies $ \epsilon^{cd} \nabla_a p_{[c|} \nabla_b q_{|d]} = 0 $, 
Equation \eqref{eq:final} thus becomes 
\be \label{eq:prefinal}
\int_\B u_a ~ d\Sigma_b {Q^{ab}} = \frac{\alpha}{2} \kappa \int_\B u^a ~ d\Sigma^b \epsilon^{xy}{}_{ab} \nabla_{[x|}(p_c \nabla_{|y]} q^c) \,.
\ee
If we apply Stokes' theorem to the right hand side of Eq.~\eqref{eq:prefinal}, our integral vanishes since the boundary of a boundary ($\mathcal{B}$ in our case) is zero. Equating the charges at the two boundaries (spatial infinity and the bifurcation surface) of our spacetime, and using Eq.~\eqref{eq:scalarcharge} and~\eqref{eq:prefinal}, it then follows that the scalar charge vanishes, namely
\begin{align}
&\int_\infty u_a ~ d \Sigma_b {Q^{ab}} = \int_\B u_a ~ d\Sigma_b {Q^{ab}}\,,
\nonumber \\ \implies 4 \pi& \mu = \frac{\alpha}{2} \kappa \int_\B u^a ~ d\Sigma^b \epsilon^{xy}{}_{ab} \nabla_{[x|}(p_c \nabla_{|y]} q^c) = 0 \,.
\end{align}

Another way to realize that the integral on the right-hand side of Eq. \eqref{eq:final} must vanish is to identify it with the Euler class of the normal bundle of $\B$ in the manifold $\cal{M}$. This is clearer when done in the Cartan formulation, as we show in Appendix \ref{Cartan}. From this, one can see that the Riemann tensor is projected onto the horizon on two indices and is orthogonal to the horizon on the other two indices. Since we are working with smooth, nowhere vanishing normals to the bifurcation surface $\B$, the Euler class of the normal bundle must \emph{vanish} \cite{BottRaoul1982DFiA}. Thus, the scalar charge cannot be sourced on the surface of a BH in a dCS spacetime devoid of a scalar potential.

\end{proof}

%%%%%%%%%%%%%%%%%%%%%%%%%%%%%%%%%%%%%%%%%%%%%%%%%%%%%%%%%%%%
%%%%%%%%%%%%%%%%%%%%%%%%%%%%%%%%%%%%%%%%%%%%%%%%%%%%%%%%%%%%
\section{Discussion}
\label{sec:conclusions}

We have here used a variational method to prove that in the dCS gravity no scalar charge, i.e.~the $1/r$ piece of the scalar field near spatial infinity, is sourced on stationary NSs or on stationary and axisymmetric BHs. This adds to the conjecture that theories with a shift-symmetric scalar field that is coupled to a topological density do not activate a scalar charge in such astrophysical objects. For the NS case, the proof relied on a continuous deformation of the spacetime away from Minkowski, while in the BH case the proof made use of generalized Noether currents generated by shift symmetry. This result is important because it implies that dipole radiation is not excited in dCS gravity when considering the inspiral phase of compact binaries, which then allows the theory to evade current binary pulsar constraints.  

%yes charge in stationary and axisymmetric BHs in dCS -> Campbell and Duncan
%no charge in stationary and spherically symmetric stars in EdGB -> Kent
%yes charge in stationary and axisymmetric BHs in EdGB -> Leo
%no charge in stationary and spherically symmetric stars in dCS -> Us
%no charge in stationary and axisymmetric BHs in dCS -> Us.  

Our work, of course, is not the first to consider the scalar charge in compact objects outside of General Relativity. Early work on BH solutions in the low-energy limit of string theory, which is related to EdGB gravity and dCS gravity, dates back to the 1990s~\cite{Campbell:1990ai,Campbell:1991kz,Campbell:1992hc,DUNCAN1992215}. However, the study of scalar charge started much more recently with the work of Yagi et al~\cite{Yagi:2015oca}, who proved analytically that the EdGB scalar charge vanishes in stationary and spherically symmetric stars. This result was recently extended by Prabhu and Stein~\cite{PhysRevD.98.021503}, who proved that the EdGB scalar charge does not vanish in stationary and axisymmetric BHs. Their argument relies in equating the scalar charge at spatial infinity with the Noether charge at the bifurcate horizon, which cannot vanish in EdGB BHs of a finite area.

Our results extend those of Yagi, et al~\cite{Yagi:2015oca} and Prabhu and Stein~\cite{PhysRevD.98.021503}, by considering both stars and BHs in dCS gravity. Using arguments similar to those of Yagi, et al~\cite{Yagi:2015oca} we first proved that the dCS scalar charge vanishes in stationary and axisymmetric stars. Then, using arguments similar to those of Prabhu and Stein~\cite{PhysRevD.98.021503} we proved that the dCS scalar charge also vanishes in stationary and axisymmetric BHs. We can then combine the results of Yagi, et al~\cite{Yagi:2015oca}, Prabhu and Stein~\cite{PhysRevD.98.021503} and ours to conclude that scalar charge is never sourced in stationary and axisymmetric stars in quadratic gravity, while it can only be sourced for stationary BHs in EdGB gravity.

One possible avenue for future work would be the investigation of the theorems discussed above through the use of a generalized version of the Chern-Gauss-Bonnet theorem to Pontryagin densities. Such a generalization does not yet exist, but if one could be developed, then it could be straightforwardly applied in a way similar to that used in~\cite{Yagi:2015oca} when considering Einstein-dilaton-Gauss-Bonnet gravity. 

Another possible extension concerns the $1/r^{2}$ piece of the scalar field at spatial infinity. Using perturbative techniques valid in the slow-rotation limit, several authors have approximations for slowly-rotating BHs, which do source a $1/r^{2}$ scalar field profile. It may be possible to express this charge as an closed integral over the bifurcation surface of a Kerr BH in the decoupling limit, just as was done in EdGB gravity for the $1/r$ part in~\cite{PhysRevD.98.021503}. Doing so may reveal how the $1/r^{2}$ charge of spinning BHs in dCS gravity depends on integrals over topological invariants.   

A third possibility is to study whether the theorems presented in this paper can be extended to non-vanishing scalar field potentials. A non-trivial potential could endow the scalar field with a mass, or it could introduce self-interactions that may lead to interesting dynamics. Such a term, however, would break shift symmetry, and this may pose a challenge to the Noether current method. Thus, it may be possible to generate interesting scalar charge effects, even at $1/r$ order, when including non-trivial potentials. 

Yet another possibility deals with dCS BHs and the concept of entropy. If we knew the exact solution for a spinning BH in dCS gravity, including its $1/r^{2}$ and higher order charges, we could then study its associated Hawking entropy. This, in turn, could be used to investigate the transition between NSs and BHs in quadratic gravity, in order to determine whether the universality found in GR~\cite{Alexander:2018wxr}~\cite{Alexander_2018} continues to hold in modified theories. 

Finally, it would be interesting to study how the $1/r^{2}$ scalar charge of a BH in dCS gravity affects the trajectory of a small compact object around a supermassive BH, and the gravitational waves that would be emitted. Such an analysis was already started in~\cite{Sopuerta:2012de,Chua:2018yng}, but the effect of dipole-dipole forces was not included then. This force, however, could enhance dCS  modifications to EMRI waveforms, allowing space-based detectors like LISA to place interesting constraints in the future.   

\section{Acknowledgments} 

We would like to thank Leo C. Stein and Kartik Prabhu for useful discussions . Nicol\'as Yunes acknowledges support from NSF grant PHY-1759615, NASA grants NNX16AB98G and 80NSSC17M0041. David Garfinkle acknowledges the support from NSF grants PHY-1505565 and PHY-1806219. Lydia Bieri acknowledges the support of NSF Grants DMS-1253149 and DMS-1811819 and Simons Fellowship in Mathematics 555809.

%%%%%%%%%%%%%%%%%%%%%%%%%%%%%%%%%%%%%%%%%%%%%%%%%%%%%%%%%%%%
%%%%%%%%%%%%%%%%%%%%%%%%%%%%%%%%%%%%%%%%%%%%%%%%%%%%%%%%%%%%

\appendix

\section{Relation between topological current and Pontryagin density}
 
\label{appendixA}

We here prove the relation
\begin{equation} \label{eq:pont}
\nabla_a K^a = \frac{1}{2} ~^*\!R R \,,
\end{equation}
where
\begin{equation}
K^a  = \epsilon^{abcd} \Gamma^n_{bm} (\partial_c \Gamma^m_{dn} + \frac{2}{3} \Gamma^m_{cl} \Gamma^l_{dn} ) ~.
\end{equation}
We take the covariant derivative of the above expression to find
\begin{equation}
\nabla_a K^a = \partial_a K^a + \Gamma^a_{ap} K^p\,,
\end{equation}
and we separate the above calculation in two parts: the partial derivative and the Christoffel connection. The partial derivative term is
\begin{align}
\label{eq:deriv}
\partial_a K^a & =  \epsilon^{abcd} (\partial_a \Gamma^n_{bm}) (\partial_c \Gamma^m_{dn}) + \frac{2}{3} \epsilon^{abcd} (\partial_a \Gamma^n_{bm}) ( \Gamma^m_{cl} \Gamma^l_{dn})
\nonumber \\
& +  \epsilon^{abcd}  \Gamma^n_{bm} (\partial_a\partial_c \Gamma^m_{dn}) + \frac{2}{3} \epsilon^{abcd}  \Gamma^n_{bm}(\partial_a  \Gamma^m_{cl}) \Gamma^l_{dn}
\nonumber \\
& + \frac{2}{3} \epsilon^{abcd}  \Gamma^n_{bm} \Gamma^m_{cl} (\partial_a \Gamma^l_{dn} )\,,
\end{align}
while the Christoffel term is
\begin{align}
\label{eq:Gamma}
\Gamma^a_{ap} K^p = \epsilon^{pbcd} \Gamma^a_{ap} \Gamma^n_{bm} (\partial_c \Gamma^m_{dn}) + \frac{2}{3} \epsilon^{pbcd} \Gamma^a_{ap} \Gamma^n_{bm} \Gamma^m_{cl} \Gamma^l_{dn}~.
\end{align}

Let us now expand the right-hand side of Eq. (\ref{eq:pont}). Using the definition of the Riemann tensor in terms of the Christoffel connection, we have
\begin{align}
\label{eq:R}
R^b{}_{dca} &= \partial_a \Gamma^b_{cd} - \partial_c \Gamma^b_{ad} + \Gamma^b_{ae}\Gamma^e_{cd} - \Gamma^b_{ce}\Gamma^e_{ad} \,,
\\
\label{eq:dualR}
~^*\! R^{dca}{}_{b} &= \frac{1}{2} \epsilon^{camn} (\partial_n \Gamma^d_{bm} - \partial_m \Gamma^d_{bn} + \Gamma^d_{ne}\Gamma^e_{mb} - \Gamma^d_{me}\Gamma^e_{nb}) ~.
\end{align}

After multiplying the terms of Eq.~(\ref{eq:R}) and (\ref{eq:dualR}), we can now compare both sides of the equation. For example, consider the first term from Eq.~(\ref{eq:R}) multiplied by the first term of Eq. (\ref{eq:dualR}). After a relabeling of the Levi-Civita tensor and using its antisymmetry property, one can easily see that this term is identical to one-fourth of the first term of Eq. (\ref{eq:deriv}). One can similarly establish that the rest of the terms produced when multiplying Eq.~(\ref{eq:R}) and (\ref{eq:dualR}) produce all the necessary terms in Eq.~(\ref{eq:deriv}) and (\ref{eq:Gamma}).

%%%%%%%%%%%%%%%%%%%%%%%%%%%%%%%%%%%%%%%%%%%%%%%%%%%%%%%%%%%%
%%%%%%%%%%%%%%%%%%%%%%%%%%%%%%%%%%%%%%%%%%%%%%%%%%%%%%%%%%%%
\section{Variation of the Pontryagin density in terms of the topological current}
\label{appendixB}
Instead of taking the variation of the product of the dual Riemann tensor and the Riemann tensor, we can consider instead using Eq. (\ref{eq:pont2}), namely
\begin{equation} \label{eq:appendixb}
\delta \int d^4 x  \sqrt{-g} \; R ~^*\!R  = 2 ~ \delta \int d^4 x  \sqrt{-g} \;  \nabla_a K^a \,,
\end{equation}
Proceeding in the same fashion as before, we can expand the Levi-Civita tensor into the Levi-Civita symbol to find
\begin{align}
  \delta \int d^4 x  \sqrt{-g} \;  \nabla_a K^a &= {\rm{sgn}}(g) \; \tilde{\epsilon}^{abcd} \int d^4 x  \;
  \nonumber \\
&\times  \delta \left[ \nabla_a \Gamma^n_{bm} \left(\partial_c \Gamma^m_{dn} + \frac{2}{3} \Gamma^m_{cl} \Gamma^l_{dn}\right) \right]\,.
\end{align}
We can expand the right-hand side using the definition of the covariant derivative to find
\begin{widetext}
\begin{align}
  \delta \int d^4 x  \sqrt{-g} \;  \nabla_a K^a &=  \partial_a \left[(\partial_c \Gamma^m_{dn} + \frac{2}{3} \Gamma^m_{cl}\Gamma^l_{dn}) (\delta \Gamma^n_{bm})\right] \tilde{\epsilon}^{abcd}
+
\partial_a \left[ -(\partial_c \Gamma^n_{bm}) \delta \Gamma^m_{dn} + \frac{2}{3} (\delta \Gamma^m_{cl} )\Gamma^n_{bm} \Gamma^l_{dn}  + \frac{2}{3} \Gamma^n_{bm} \Gamma^m_{cl} (\delta \Gamma^l_{dn})  \right] \tilde{\epsilon}^{abcd}
\nonumber \\
&+
\delta \Gamma^a_{ap}\left[ (\partial_c \Gamma^m_{dn}) + \frac{2}{3} \Gamma^m_{cl}\Gamma^l_{dn}  \right] \Gamma^n_{bm} \tilde{\epsilon}^{pbcd}
+
\Gamma^a_{ap} (\partial_c  \Gamma^m_{dn} + \frac{2}{3} \Gamma^m_{cl}\Gamma^l_{dn} ) \delta \Gamma^n_{bm}\tilde{\epsilon}^{pbcd}
\nonumber \\
&+
   \Gamma^a_{ap} \left[ - \partial_c \Gamma^n_{bm} ( \delta \Gamma^m_{dn} ) + \frac{2}{3} \Gamma^n_{bm} \Gamma^l_{dn} (\delta \Gamma^m_{cl})  + \frac{2}{3} \Gamma^n_{bm}  (\delta \Gamma^l_{dn} ) \Gamma^m_{cl} \right] \tilde{\epsilon}^{pbcd} \,,
\end{align}
\end{widetext}
Expanding these equations and relabelling indices, returns the right hand side of Eq. (\ref{eq:appendixb}), namely
\begin{align}
&{\rm{sgn}}(g) \; \tilde{\epsilon}^{abcd} \int d^4 x  \; \delta \left[ \nabla_a \Gamma^n_{bm} \left(\partial_c \Gamma^m_{dn} + \frac{2}{3} \Gamma^m_{cl} \Gamma^l_{dn} \right) \right]
\\
&= {\rm{sgn}}(g) \; \tilde{\epsilon}^{abcd} \int d^4 x \;  \nabla_a (R^m_{dcn} \delta \Gamma^n_{bm})\,.
\end{align}

%%%%%%%%%%%%%%%%%%%%%%%%%%%%%%%%%%%%%%%%%%%%%%%%%%%%%%%%%%%%%%%%%%%%%%%%%%%%%%%%%%%%%%%%%%%%%%%%%%%%%%%%%%%%%%%%%%%%%%%%%%%%%%%%%%%%

\section{A Cartan formulation}
\label{Cartan}
One can also use the Cartan formulation to reformulation Sec.~\ref{sec:IV}. The process is very similar to the tensorial formulation. We begin by expressing the Lagrangian as 
\be
\mathcal{L} = \mathcal{L}_g + \mathcal{L}_\vartheta \,,
\ee
where the Lagrangian is again an n-form which is a collection of dynamical fields. We know that
\be
\mathcal{L}_\vartheta = - \frac{1}{2}(*d\vartheta)\wedge d\vartheta + \frac{\alpha}{4} {}^* R R \,,
\ee
where $*$ is the Hodge dual defined as 
\be
(*Z)_{a_1 ... a_{4-p}} = \frac{1}{p!} \epsilon^{b_1 ... b_{p}}{}_{a_1 ... a_{4-p}} Z_{b_{1} ... b_p} \,,
\ee
and $d$ is the exterior derivative. Rewriting the Pontryagin term as 
\begin{align}
{}^* RR &= \frac{1}{2} \epsilon^{cdef} R^a{}_{bcd} R^b{}_{aef} \nonumber \\
&= R^{a}{}_b \wedge R^b{} _{a} \\
&= d (\omega^a{}_b \wedge (R^b{}_a - \frac{1}{3} \omega^b{}_c \wedge \omega^c{}_a )) \,,
\end{align}
where $\omega^{ab} $ is the connection one-form.

Varying the Lagrangian by the least action principle in turn gives us the equations of motion, namely
\be
\delta \mathcal{L} = E_g \delta g + E_\vartheta \delta \vartheta + d\theta \,,
\ee
where we have defined
\be
E_\vartheta = 0 = \Box \vartheta + \frac{\alpha}{4} {}^* R R = d {}*d\vartheta + \frac{\alpha}{4} {}^* RR \,,
\ee
and $\theta$ is the symplectic potential arising due to the variation of the Lagrangian and represents the boundary terms. This is a local (n-1)-form and a covariant functional of $(g_{\mu \nu}, \vartheta) $, so it depends on $(\delta g_{\mu \nu}, \delta \vartheta )$.  This functional can be expressed as 
\be
\theta = \theta_g(\delta g) - *d\vartheta \delta \vartheta \,.
\ee

Assuming the Lagrangian is not invariant under Lorentz diffeomorphisms, we apply a transformation such that $ \delta_c g_{\mu \nu} = 0 , \delta_c \vartheta = c $ for $c$ a constant. Under this transformation, the symplectic potential takes the form
\be 
\theta = -c * d \vartheta \,,
\ee
and thus the Lagrangian changes by 
\be
\delta_c \mathcal{L} = c \frac{\alpha}{4} {}^* RR = d(c \frac{\alpha}{4} \gamma)\,.
\ee 
This can be thought of as a shift symmetry, whose corresponding (n-1)-form Noether current is 
\be
\J = - * d\vartheta - \frac{\alpha}{4} \gamma
\ee
Notice then that by construction $ d \J =  -E_\vartheta $, and thus, $\J$ is indeed a conserved current on shell ie. $d\J = 0 $ on shell. Because $\J$ is an exact-closed form on shell, we then have that 
\be
\J = d \q \,,
\ee
where $\q$ is the conserved charge and an (n-2)-form on shell. 

We now define the covariant shift current for some vector field $X$ using the same reasoning as that of Sec.~\ref{sec:IV}. The covariant conserved charge of that section is related to the following covariant shift current 
\begin{align}
\mathcal{J} &= - \pounds_X \J - \frac{\alpha}{4} \pounds_X \gamma + \frac{\alpha}{2} R^b{}_a \wedge \pounds_X \omega^a{}_b \nonumber \\
&= \pounds_X(*d \vartheta ) + \frac{\alpha}{2} R^b{}_a \wedge \pounds_X \omega^a{}_b \,.
\end{align}
The Lie derivative of the connection with respect to the vector fields on the bundle is 
\be \label{eq:1}
\pounds_X \omega^a{}_b = X . R^a{}_b + D (X.\omega^a{}_b) \,.
\ee
For a bundle vector field that is Lie-dragged along the tetrad $e^a $, i.e.~$ \pounds_X e^a = 0 $, we then also have that
\be \label{eq:2}
X.\omega^{ab} = - e^a_\mu e^b_\nu \nabla^\mu X^\nu \,.
\ee 
An in-depth discussion regarding this can be found in \cite{PhysRevD.98.021503} [Also see \cite{PhysRevD.92.124010}]. The term on the left-hand side of Eq.~\eqref{eq:2} is the interior product, generally defined by
\be
(X.Z)_{ab...m} = X^\mu Z_{\mu ab .... m}\,.
\ee

With this at hand, we can now write the covariant shift current using Eq. \eqref{eq:1} as 
\be
\mathcal{J} = d\q_X + X.E_\vartheta\,,
\ee
with 
\be
\q_X = X.(*d\vartheta) + \frac{\alpha}{2} R^a{}_b (X. \omega^b{}_a) \,.
\ee
If $ X^a $ is a Killing field, we can use Eq. \eqref{eq:2} to find that the covariant shift charge is
\be \label{eq:3}
\q_X = X.(*d\vartheta) - \frac{\alpha}{2} R^{ab} e^\mu _a e^\nu_b \nabla_\mu X_\nu \,.
\ee 

Let us now consider the asymptotically flat, stationary-axisymmetric BH spacetime of Fig.~\ref{Fig. 1}, with $\B$ and $\Sigma $ as defined earlier. We then have that $X^a = L^a $ where $L^a$ is the Killing field, and integrating $ 0 = \mathcal{J} = d \q_L $ we now have
\be \label{eq:4}
\int_\infty \q_L = \int_\B \q_L\,.
\ee
We integral on the left-hand side is to be evaluated over an asymptotic 2-sphere of radius r in the limit $ r \to \infty $. We define $\epsilon_2 $ to be the induced area element of this 2-sphere. We then have that
\be \label{eq:5}
\int_\infty \q_L = \int_\infty L .(*d \vartheta) = -\int_\infty \epsilon_2 \partial_r \vartheta = 4 \pi \mu \,.
\ee

Let us now focus on the right-hand side of Eq.~\eqref{eq:4}. On $\B$, we have $ L^a|_\B = 0$ and $ \nabla_a L_b |_\B= \kappa \tilde{n}_{cd} $ with $\tilde{n}_{cd} $ the binormal to the bifurcation surface \cite{Wald:1984cw} and $\kappa $ a constant on the bifurcation surface. Using this, the charge integral at the bifurcation surface is given by
\be \label{eq:6}
\int_\B \q_L = - \frac{\alpha}{2} \kappa \int_\B \tilde{n}_{ab} R^{ab} = 0
\ee
The above integral vanishes because the term in the integral is the Euler class of the normal bundle of $\B$ in M. Since smooth and nowhere vanishing normals to $\B$ are needed, the Euler class must vanish. Thus, Eq. \eqref{eq:5} \eqref{eq:6} along with Eq.~\eqref{eq:4}, we find our desired result: \textit{no scalar charge is sourced in stationary, axisymmetric BH spacetimes that possess a Killing horizon in dCS gravity.}

\bibliography{ppr_v}

\end{document}